\documentclass[runningheads]{llncs}
\usepackage[T1]{fontenc}
\usepackage{graphicx}
\usepackage{subcaption} 
\usepackage{tikz}
\usetikzlibrary{spy}
\usepackage{amsmath}
\usepackage{amssymb}
\usepackage{bm}

\newcommand{\R}{\mathbb{R}}
\newcommand{\grad}{\bm{\nabla}}
\newcommand{\Div}{\mbox{\rm{\bf{div\,}}}}

\begin{document}
\title{Generalised Scale-Space Properties for Probabilistic Diffusion Models}
\titlerunning{Generalised Scale-Space Properties for Probabilistic Diffusion 
Models}
%

\author{Pascal Peter
}
\authorrunning{Pascal Peter}
\institute{Mathematical Image Analysis Group,
    Faculty of Mathematics and Computer Science,\\ Campus E1.7,
    Saarland University, 66041 Saarbr\"ucken, Germany.\\
    peter@mia.uni-saarland.de}
\maketitle              
\begin{abstract}

Probabilistic diffusion models enjoy increasing popularity in the 
deep learning community. They generate convincing samples from a 
learned distribution of input images with a wide field of practical 
applications. Originally, these approaches were motivated from drift-diffusion
processes, but these origins find less attention in recent, practice-oriented 
publications.

We investigate probabilistic diffusion models from the viewpoint of 
scale-space research and show that they fulfil generalised scale-space 
properties on evolving probability distributions. Moreover, we discuss 
similarities and differences between interpretations of the physical core 
concept of drift-diffusion in the deep learning and model-based world. To this 
end, we examine relations of probabilistic diffusion to osmosis filters. 

\keywords{probabilistic diffusion \and scale-spaces \and drift-diffusion \and  
osmosis.}
\end{abstract}
\section{Introduction}

Probabilistic diffusion models introduced by 
Sohl-Dickstein et al.~\cite{SMDH13}  are enjoying a rapid rise in 
popularity \cite{HJA20,KSBH21,RBLE+22,SDME21} fueled by the publicly available 
stable diffusion framework of Rombach et al.~\cite{RBLE+22}. These deep learning 
models are 
generative in nature: Given a random seed, they can create new samples that fit 
to a given set of training data, for instance a certain class of images. 
Especially the numerous excellent text-to-image results based on stable 
diffusion \cite{RBLE+22} resonate not only with the scientific 
community, sparking many recent publications, but also with the general public.

Due to their tremendous success in practical applications, the roots of these 
approaches have received less attention than their efficient implementation by 
deep neural networks. While current interpretations often consider 
probabilistic diffusion as highly sophisticated versions of denoising 
autoencoders, its original roots lie in well-known physical processes, namely 
drift-diffusion. Diffusion processes have been closely investigated by the 
scale-space community \cite{AGLM93,Ii62,SchW98,We97}, and drift-diffusion has 
model-based 
applications in the form of osmosis filtering proposed by Weickert et 
al.~\cite{WHBV13}. We aim to bring the scale-space and deep learning 
communities closer together by showing that there are generalised 
scale-space concepts behind one of the most successful current paradigms in 
deep learning.
\medskip

\textbf{Our Contribution.} We establish the first 
generalised scale-space interpretation for probabilistic diffusion. Compared to 
traditional scale-spaces, it describes the gradual simplification of 
probability distributions towards a non-flat steady state which does not carry 
information of the initial distribution anymore. We introduce entropy-based 
Lyapunov sequences and establish invariance statements. Moreover, we discuss 
relations to deterministic osmosis filters, highlighting similarities and 
differences.

\medskip

\textbf{Related Work.} Probabilistic diffusion models were introduced 
by Sohl-Dickstein et al.~\cite{SMDH13} as an alternative to existing generative 
neural networks such as generative adversarial networks \cite{GPMX14}. They use deep 
learning to invert a Markov process 
that gradually adds noise to an image. This allows to generate new samples from 
the distribution of the training data. Beyond the initial applications such as 
text-to-image, superresolution, and inpainting, many improvements and new 
practical uses have been proposed (e.g. \cite{HJA20,KSBH21,SDME21,SMDH13}).  
With their publicly available stable diffusion model including trained weights, 
Rombach et al.~\cite{RBLE+22} have unleashed a torrent of real-world 
applications for the initial concept. Due to the tremendous research interest 
in the topic, a full review is beyond the scope of this paper.

We investigate probabilistic diffusion from a new scale-space perspective. 
Classical scale-spaces embed images into a family of systematically simplified 
versions based on partial differential equations (PDEs) 
\cite{AGLM93,Ii62,Li11,SchW98,We97} or pseudo-differential operators 
\cite{DFGH04,SW16}. 

Stochastic scale-spaces are fairly rare. Conceptually, the Ph.D. thesis of 
Majer~\cite{Ma00} comes closest to our own considerations 
since it also deals with stochastic simplification and drift-diffusion. However,
it is unrelated to deep learning and shuffles image pixels to 
remove information. A similar, local shuffling has been proposed by Koenderink 
and Van Doorn~\cite{KV99} under the name ``locally orderless images''. Stochastic 
considerations related to scale-spaces have been made 
w.r.t. the statistics of natural images reflected by Gaussian image models 
\cite{Pe03} and practical applications in stem cell differentiation 
\cite{HKMR+16}.

Notably, probabilistic diffusion was originally motivated from drift-diffusion 
processes which can be described by the Fokker-Planck equation \cite{Ri84}. 
Osmosis filters for visual computing, a generalisation of diffusion filtering 
\cite{We97}, have been derived from the same concept. The corresponding 
continuous theory was proposed by Weickert et al.~\cite{WHBV13}, while 
Vogel et al.~\cite{VHWS13} provide results in the discrete setting, and 
additional properties were shown by Schmidt~\cite{Sc18}. We discuss these in 
more detail in Section~\ref{sec:osmosis}. Osmosis has proven particularly 
useful for image editing \cite{DMM18,VHWS13,WHBV13},  
shadow removal \cite{DMM18,PCCSW19,WHBV13}, and image fusion~\cite{PCBP+20}. 

This implies connections to other fields of research. For instance, Sochen has 
established relationships between drift-diffusion and the Beltrami flow 
\cite{So01b}. Hagemann et al.~\cite{SHH23} have proposed a general 
framework 
that ties together many concepts, including probabilistic diffusion, under the 
common model of normalising flows.

\medskip

\textbf{Organisation of the Paper.}
In Section~\ref{sec:probdiff} we introduce probabilistic diffusion in its 
formulation as a Markov process and propose a corresponding generalised 
scale-space theory in Section~\ref{sec:probdiffscale}. After a discussion of 
similarities and differences to osmosis filtering in 
Section~\ref{sec:discussion}, we draw conclusions and assess 
potential future benefits of this connection in Section~\ref{sec:conclusion}.


\section{Probabilistic Diffusion}
\label{sec:probdiff}

Probabilistic diffusion~\cite{SMDH13} differs from most classical filters 
associated with scale-spaces. Instead of a single initial image, it considers a 
set of known images. These training data act as samples 
for common statistics that are not known directly. We can interpret discrete
training images $\bm f_1,...,\bm f_{n_t} \in \R^{n_x n_y n_c}$ with $n_c$ 
colour channels of size $n_x \times n_y$ as 
realisations of a random variable $\bm F$. The unknown \emph{target 
distribution} is expressed by its probability density function $p(\bm F)$.

Probabilistic diffusion maps this unknown $p(\bm F)$ to a simple, well-known 
distribution such as the standard normal distribution (i.e. Gaussian noise).
Practical applications exploit that we can also map samples from the noise 
distribution back to the unknown distribution $p(\bm F)$ (see
Section~\ref{sec:backwards}). This makes the model generative, since it can 
create new images that resemble the training data. 

However, we focus first on the \emph{forward process} and show that it 
constitutes a scale-space in Section~\ref{sec:probdiffscale}.
Its evolution is described by a time-dependent random variable $\bm U(t)$. At 
time $t=0$ it has the same distribution as our training data, i.e. $p(\bm 
F)$. A sequence of $m$ temporal realisations $\bm u_1, ..., \bm u_m$ at times 
$t_1 < t_2 < ... < t_m$ is referred to as one possible \emph{trajectory} of 
$\bm U$. In a Markov 
process \cite{Ga85}, $\bm u_i$ 
only depends on $\bm u_{i-1}$ and not on the previous trajectory $\bm u_0, ..., 
\bm u_{i-2}$.  In terms of conditional 
transition probabilities, the Markov property is formulated as
\begin{align}
    \label{eq:markov_property}
    p(\bm u_i | \bm u_{i-1}, ..., \bm u_0) \;=\; p(\bm u_i | \bm u_{i-1}) \, . 
\end{align}
This notation refers to the probability of the random variable $\bm 
U(t)$ assuming value $\bm u_i$ at time $t_i$, given that we observed $\bm 
U(t_{i-1})=\bm u_{i-1}$.
Due to the Markov property, the probability density of the whole trajectory
can be successively traced back to the distribution $p(\bm u_0)=p(\bm F)$ of 
the training data according to
\begin{align}
    \label{eq:density_fwdtrajectory}
p(\bm u_0, ..., \bm u_m) \;=\; p(\bm u_0) \prod_{i=1}^m  p(\bm u_i | \bm 
u_{i-1}) 
\, .
\end{align}
More concretely, we consider Gaussian transition probabilities 
\begin{align}
    \label{eq:forwardtransition}
    p(\bm u_i | \bm  u_{i-1}) \;=\; \mathcal{N}\left(\sqrt{1-\beta_i} \, \bm 
    u_{i-1}, \,
    \beta_i \bm I \right)\, . 
\end{align}
Here, $\mathcal{N}(\bm \mu, \bm \sigma)$ denotes a multivariate Gaussian 
with unit matrix $\bm I \in \R^{n \times n}$ where $n=n_x n_y n_c$ is the 
number of pixels. Since the covariance matrix is diagonal, this corresponds 
to independent, identically distributed Gaussian noise with mean $\mu_{i,j} = 
\sqrt{1-\beta_i} u_{i-1,j}$ and standard deviation $\sigma_i = 
\sqrt{\beta_i}$ 
for each 
pixel $j$. The parameters $\beta_i \in (0,1)$ can either be user-specified or 
learned.
For our following considerations, it is useful to  
express the random variable at a time $t_i$ in terms of the random 
variable at the previous time $t_{i-1}$ and Gaussian noise $\bm G$ from the 
standard 
normal distribution $\mathcal{N}(\bm 0, \bm I)$. 
Eq.~\eqref{eq:density_fwdtrajectory} directly implies that
\begin{align}
    \label{eq:forwardrandomvar}
    \bm U_i \;=\; \sqrt{1-\beta_i} \, \bm U_{i-1} + \sqrt{\beta_i}  \, \bm G \, 
    .
\end{align}
With Eq.~\eqref{eq:forwardtransition}, a trajectory of 
images can be obtained from a starting image $\bm u_0 = \bm f$ by rescaling the 
image and adding a realisation of Gaussian noise in each step. 
As in \cite{HJA20}, we can also specify the 
transition from time $0$ to time $t_i$ as
\begin{align}
    \label{eq:starttot}
    p(\bm u_i | \bm  u_0) \;=\; \mathcal{N} \left(\sqrt{\prod_{j=0}^i (1 - 
        \beta_j)} \, 
    \bm u_0, \, \bm I - \prod_{j=0}^i (1-\beta_j) \bm I \right)\, . 
\end{align}
With these insights into the forward process, we are suitably equipped to 
establish a scale-space theory for probabilistic diffusion.


\section{Generalised Probabilistic Diffusion Scale-Space}
\label{sec:probdiffscale}

In the following, we will consider a scale-space in the sense of probability 
distributions. As such, we do not specify properties of individual images, 
but of the marginal densities $p(\bm u_i)$ instead. 

\subsection{Generalised Scale-Space Properties}
\label{sec:properties}


\medskip
\noindent
{\bf Property 1: Training Data Distribution as Initial State.} 
By definition, the distribution $p(\bm u_0)$ at time $t_0=0$ is identical to 
the distribution $p(\bm F)$ of the training data.


\medskip
\noindent
{\bf Property 2: Semigroup Property.} 
We can acquire $p(\bm u_i)$ equivalently in $i$ steps from $p(\bm u_0)$ or in 
$\ell$ steps from $p(\bm u_{i-\ell})$. 
\begin{proof}
This follows from  the recursive definition of 
the probability density \eqref{eq:density_fwdtrajectory}. To find the 
distribution of an individual step in the trajectory, we integrate over all 
possible paths that lead to 
$\bm 
u_i$ and consider the marginal probability density 
\begin{align}
     p(\bm u_i) \;=\; \int p(\bm u_0, ..., \bm u_i) \,  d \bm u_0 \cdots d \bm 
     u_{i-1} \, . 
\end{align}
Using the definition of the joint probability density of the Markov process, we
obtain the aforementioned two alternative ways to express $p(\bm u_{i})$:
\begin{align}
    p(\bm u_i) &\;=\;  \int p(\bm u_0) \prod_{j=1}^{i} p(\bm 
    u_j |\bm  u_{j-1}) \,  d \bm u_0 \cdots d \bm 
    u_{i-1} \, \\
    &\;=\;  \int p(\bm u_{i-\ell}) \prod_{j=i-\ell+1}^{i-1} p(\bm 
    u_j |\bm  u_{j-1}) \,  d \bm u_{i-\ell} \cdots d \bm 
    u_{i-1} \, . 
\end{align}
Due to the Markov property \eqref{eq:markov_property}, we can 
start the trajectory at any intermediate time $i-\ell$. \qed
\end{proof}


\medskip
\noindent
{\bf Property 3: Lyapunov Sequences.} 
 In classical scale-spaces (e.g. with diffusion),  Lyapunov sequences quantify 
 the change in the evolving image with increasing scale parameter. They 
 constitute a  measure of image simplification~\cite{We97} in terms of monotonic 
 functions. In practice, they often represent the information content of an image at a 
 given scale. Here, we define a Lyapunov sequence on the evolving probability density 
 instead.
Our first Lyapunov sequence, the differential entropy, indicates that the 
distribution of $\bm U$ gradually becomes more random with increasing time 
$t_i$. 
\begin{proposition}[Increasing Differential Entropy] 
   The differential entropy
   \begin{equation}
       \label{eq:diffentropy}
       H(\bm U_i) \;:=\; -\int  p(\bm u_i) \ln p(\bm u_i) d\bm u_i
   \end{equation}
   increases with $t_i$ under the following assumptions 
   for $\beta_i$:
   \begin{equation}
       \label{eq:beta_limits}
       \frac 1 2 - \sqrt{\frac 1 4 - \frac{1}{(2 \pi e)^n} } \;\leq\; \beta_i 
       \;\leq\;  
       \frac 1 2 + \sqrt{\frac 1 4 - \frac{1}{(2 \pi e)^n} } \, .
   \end{equation}
   Here, $n=n_x n_y n_c$ denotes the total number of pixels.
\end{proposition}
\begin{proof} According to Eq.~\eqref{eq:forwardrandomvar}, we can rewrite the 
    entropy at $i+1$ as
    \begin{align}
        H(\bm U_{i+1}) &\;=\; H\left(\sqrt{1-\beta_i} \, \bm U_i + 
        \sqrt{\beta_i} 
        \, \bm 
        G\right) \\
        &\;=\;
        H(\bm U_i) + \ln\left(\sqrt{1-\beta_i}\right) + H(\bm G) + 
        \ln\left(\sqrt{\beta_i}\right) \, .
    \end{align} 
    Here we have used that the Gaussian noise does not depend on the images in 
    the 
    time step, and thus $\bm U_i$ and $\bm G$ are independent. Therefore, the 
    entropy can be additively decomposed. Consequentially, the differential 
    entropy 
    is monotonously increasing if
    \begin{equation}
        \label{eq:entropydiff}
        H(\bm G) + \ln\left(\sqrt{(1-\beta_i) \beta_i}\right) \;\geq\; 0 \, .
    \end{equation}
    This holds under restrictions for $\beta_i$:
    \begin{align}
        &\ln\left(\sqrt{\beta_i (1-\beta_i)}\right) \;\geq\; -H(\mathcal{N}(\bm 
        0,\bm I)) 
        &\Leftrightarrow \quad&  \sqrt{\beta_i (1-\beta_i)} \;\geq\; (2\pi 
        e)^{- 
        \frac 
            n 2} 
        \\
        \Leftrightarrow \quad &  \beta_i (1-\beta_i) \;\geq\; (2 \pi e)^{-n}   
        &\Leftrightarrow \quad&  \beta_i^2 - \beta_i + (2 \pi e)^{-n} \;\leq\; 
        0 \, .
    \end{align}
    Standard rules for quadratic functions yield the conditions in 
    Eq.~\eqref{eq:beta_limits} for which the inequality is fulfilled. Note that 
    for $n \rightarrow \infty$, the lower limit goes to $0$ and the 
    upper limit to $1$. In practice, this holds for reasonable choices of 
    $\beta_i$ that are not too close to the boundaries of its range $(0,1)$. 
    \qed
\end{proof}

Alternatively, we can also consider the increasing conditional entropy given 
the distribution of the training data. Intuitively, this means that more 
information is needed to describe $\bm U_i$ 
given $\bm U_0$ with increasing $t_i$, which reflects that the initial 
information is gradually destroyed by  noise.

\begin{proposition}[Increasing Conditional Entropy]  The conditional entropy
    \begin{align}
        \label{eq:condentropy}
        H(\bm U_i | \bm U_0) &\;=\; - \int \int p(\bm u_i, \bm u_0) \ln p(\bm 
        u_i 
        | 
        \bm u_0) \, d \bm u_0 d \bm u_i  \, .
    \end{align}
    increases with $t_i$ for all $\beta_i \in (0,1)$.
\end{proposition}
\begin{proof}
    \begin{align}
        H(\bm U_i | \bm U_0)
        &\;=\; \int p(\bm u_0)  - \int p(\bm u_i | \bm u_0) \ln p(\bm  u_i | 
        \bm 
        u_0) 
    d \bm u_i d \bm  u_0   \\
    & \;=\; \int  p(\bm u_0) \ln\Bigg((2\pi e)^n \Bigg(1- \prod_{j=1}^i
    (1-\beta_j)\Bigg)^n\Bigg) d \bm u_0\\
    &\;\geq\;  \int  p(\bm u_0) \ln\Bigg( (2\pi e)^n \Big(1- 
    \prod_{j=1}^{i-1} 
    (1-\beta_j)\Bigg)^n\Bigg) d \bm u_0 \, .\\
    & \; = \; H(\bm U_{i-1} | \bm U_0)
\end{align}
According to Eq.~\eqref{eq:starttot}, $ p(\bm u_i | \bm u_0)$ is a normal  
distribution, and the inner integral is its entropy. It only depends on the  
covariance of the normal distribution and is thus independent of $\bm u_0$. \qed
\end{proof}


\medskip
\noindent
{\bf Property 4: Permutation Invariance.} 
Consider an arbitrary permutation function $P(\bm f)$ that reorders the pixels 
of an image $\bm f$ from the initial distribution. Note that arbitrary 
permutations specifically also include translations and rotations by $90^\circ$ 
increments. 
Probabilistic diffusion is invariant under permutations in the sense that 
trajectories are also permuted accordingly, which corresponds to classical 
invariances on individual images.

\begin{proposition}[Permutation Invariant Trajectories] 
       Any trajectory\\ $\bm v_1, ... \bm v_m$ starting from the permuted 
       initial image $\bm v_0 := P(\bm f)$ can be obtained by the same 
       permutation $P$ from a trajectory $\bm u_0, ... \bm u_m$ starting  
       from the original  image $\bm f =: \bm u_0$, i.e. for all $i$, we have 
       $\bm v_i 
       = P(\bm u_i)$, and vice versa.
\end{proposition}
\begin{proof}
Let $\bm g_i$ denote the Gaussian noise realisation of $\bm G$ from 
Eq.~\eqref{eq:forwardrandomvar} that occurs in the transition from $\bm 
v_{i-1}$ to $\bm v_i$. We define the transition noise from $\bm u_{i-1}$ to 
$\bm u_i$ as $\tilde{\bm g}_i := P^{-1}(\bm g_i)$, which is also from 
$\mathcal{N}(\bm 0, \bm I)$. Now we can prove the claim by induction, starting 
with $\bm v_0 = P(\bm u_0)$. Assuming that $\bm v_{i-1} = P(\bm 
u_{i-1})$, we obtain  
\begin{equation}
    \bm v_i \;=\; \sqrt{1-\beta_i} \bm v_{i-1} + \sqrt{\beta_i} \bm g_i  \;=\; 
    \sqrt{1-\beta_i} P(\bm u_{i-1}) + \sqrt{\beta_i} P(\tilde{\bm g}_i) \;=\; 
    P(\bm u_i) \, .
\end{equation}
Since $P$ is a bijection, we have a one-to-one mapping between all possible 
trajectories from a permuted image and the permuted trajectories of the 
original.\qed 
\end{proof}
The proof above implies that permuting the initial data leads to the same 
permutation of the corresponding trajectories.

\medskip
\noindent
{\bf Property 5: Steady State.} 
The steady state distribution for $i \rightarrow \infty$ is a multivariate 
Gaussian distribution $\mathcal{N}(\bm 0, \bm I)$ with mean $\bm 0$ and standard 
deviation~$\bm I$. This is an effect that results directly from adding 
Gaussian noise in every step of the Markov process and has been used by Sohl-Dickstein 
et al.~\cite{SMDH13} without proof. In the following we provide a short formal 
argument for the sake of completeness.

\begin{proposition}[Convergence to the Standard Normal Distribution] 
    With the assumptions on $\beta_i$ from Property 3, the forward process described 
    by Eq.~\eqref{eq:density_fwdtrajectory} 
    converges 
    to the standard normal distribution for $i \rightarrow \infty$.
\end{proposition}
\begin{proof}
    The statement follows from 
    Eq.~\eqref{eq:starttot} according to
    \begin{align}
        \bm u_i \;=\; \underbrace{\sqrt{\prod_{j=0}^i 
        (1-\beta_j)}}_{\stackrel{i 
                \to
                \infty}{\to 0}} \, \bm u_0 + 
        \underbrace{\sqrt{1-\prod_{j=0}^{i} 
                (1-\beta_j)}}_{\stackrel{i \to \infty}{\to 1}} 
               \, \bm 
        \xi 
    \end{align}
    with $\bm \xi$ from $\mathcal{N}(\bm 0, \bm I)$. \qed
\end{proof}

Note that classical scale-spaces, e.g. those resulting from diffusion 
processes~\cite{We97}, converge to a flat image as the state of least 
information. 
However, 
in a generalised setting, probabilistic diffusion still follows the scale-space 
idea of gradual simplification by systematic removal of information. 

\subsection{PDE Formulation and Reverse Process}
\label{sec:backwards}

In Section~\ref{sec:properties}, we have established that probabilistic 
diffusion in its 
Markov formulation constitutes a generalised scale-space. Like many existing 
scale-spaces, probabilistic diffusion can also be expressed in a PDE 
formulation. 
Feller~\cite{Fe49} has shown a connection of a Markov process of type 
\eqref{eq:density_fwdtrajectory} if
the stochastic moments
\begin{align}
    m_k(\bm u_t, t) \;=\; \lim_{h \rightarrow 0} \frac{1}{h} \int p(\bm u_{t+h} 
    , 
    \bm u_t) (\bm u_{t+h} - \bm u_t)^k d \bm u_{t+h}
\end{align}
exist for $k\in \{1,2\}$. In this case the probability density $p(\bm u_\tau, 
\bm u_t)$ with $\tau < t$ is a solution to the partial differential equation
\begin{align}
    \label{eq:fellerfwd}
    \frac{\partial}{\partial t} p \;=\; \frac{1}{2} 
    \frac{\partial}{\partial \bm u_t \partial \bm u_t} \left(m_2(\bm u_t, t) 
    p\right) 
    + 
    \frac{\partial}{\partial \bm 
    u_t} \left( m_1(\bm u_t, t) p \right) \, .
\end{align}
This is a drift-diffusion equation with a drift coefficient $m_1$. 

A crucial component for the success of probabilistic diffusion is the 
counterpart of the aforementioned forward process, the \emph{backward} process. 
Feller~\cite{Fe49} found that if a solution to the PDE \eqref{eq:fellerfwd} 
exists, then it also solves the backward equation
\begin{align}
    \label{eq:fellerback}
    \frac{\partial}{\partial \tau} p \;=\; \frac{1}{2} m_2(\bm u_\tau, \tau) 
    \frac{\partial}{\partial \bm u_\tau \partial \bm u_\tau } p + 
    m_1(\bm u_\tau, \tau) \frac{\partial}{\partial \bm 
            u_\tau}  p \, .
\end{align}
Note that this equation is formulated w.r.t. the earlier time $\tau$, thus 
yielding a backward perspective where transitions from $t$ to $\tau$ are 
considered.
Sohl-Dickstein et al.~\cite{SMDH13} use the fact that this reverse 
process has a very similar form compared to the forward process. It starts with 
the Gaussian noise distribution $\mathcal{N}(\bm 0, \bm I)$ and converges to 
$p(\bm F)$. However, in contrast to the forward process, the mean and standard 
deviation of the Gaussian transition probabilities are unknown. These 
parameters are learned with a neural network such that the steady state 
minimises the cross entropy to $p(\bm F)$. 
Details on how to find the reverse process have been the topic of many 
publications and have been refined considerably compared to the original 
publication of Sohl-Dickstein et al.~\cite{SMDH13}. Since this is not the focus 
of our work, we refer to \cite{HJA20,KSBH21,RBLE+22,SDME21} for more details. 

The reverse process can sample new images from the distribution 
$p(\bm F)$ that are not part of the training data. Additionally, this 
probability distribution can be conditioned with side information. Providing a 
textual description of the image content specifies the sampling, thus creating 
a text-to-image approach~\cite{RBLE+22}. Similarly, inpainting can be implemented by 
using known image parts as side information~\cite{RBLE+22,SMDH13}.


\begin{figure}[!h]
    \begin{center}
    \newlength{\imgwidth}
    \setlength{\imgwidth}{0.19\textwidth}
    \begin{tabular}{ccccc}
        \multicolumn{5}{c}{(a) probabilistic diffusion}\\[0.5mm]
        \includegraphics[width=\imgwidth]{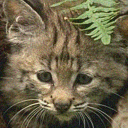}
        &
        \includegraphics[width=\imgwidth]{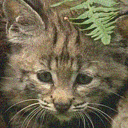}
        &
        \includegraphics[width=\imgwidth]{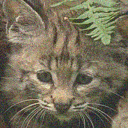}
        &
        \includegraphics[width=\imgwidth]{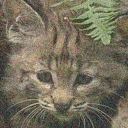}
        &
        \includegraphics[width=\imgwidth]{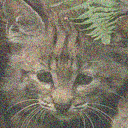}
        \\[-0.5mm]
        $i=0$ & $i=1$ & $i=2$ & $i=4$ & $i=8$ \\[0.25mm]
        \includegraphics[width=\imgwidth]{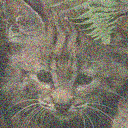}
        &
        \includegraphics[width=\imgwidth]{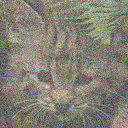}
        &
        \includegraphics[width=\imgwidth]{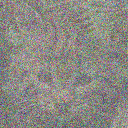}
        &
        \includegraphics[width=\imgwidth]{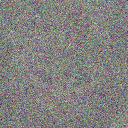}
        &
        \includegraphics[width=\imgwidth]{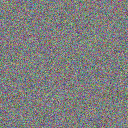}
        \\[-0.5mm]
        $i=32$ & $i=128$ & $i=512$ & $i=2048$ & $i=8192$ \\[0.25mm]
        
        \hline\\[-2mm]  
        \multicolumn{5}{c}{(b) osmosis}\\[0.5mm]
        \includegraphics[width=\imgwidth]{images/original.png}
        &
        \includegraphics[width=\imgwidth]{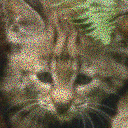}
        &
        \includegraphics[width=\imgwidth]{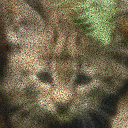}
        &
        \includegraphics[width=\imgwidth]{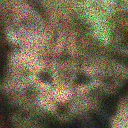}
        &
        \includegraphics[width=\imgwidth]{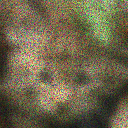}

       \\[-0.5mm]
$i=0$ & $i=1$ & $i=2$ & $i=4$ & $i=8$ \\[0.25mm]

        \includegraphics[width=\imgwidth]{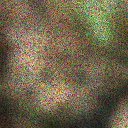}
        &
        \includegraphics[width=\imgwidth]{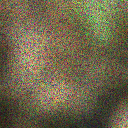}
        &
        \includegraphics[width=\imgwidth]{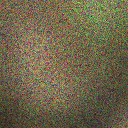}
        &
        \includegraphics[width=\imgwidth]{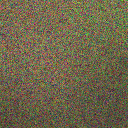}
        &
        \includegraphics[width=\imgwidth]{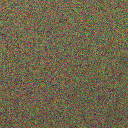}
\\[-0.5mm]
$i=32$ & $i=128$ & $i=512$ & $i=2048$ & $i=8192$ \\[0.25mm]

    \end{tabular}
\end{center}%
\vspace{-5mm}
    \caption{\label{fig:scalespaces} The probabilistic diffusion trajectory (a)
    was obtained with $\beta_t=0.02$ and $i$ indicates the number of steps. The 
    osmosis evolution (b)  uses time 
    step size  $\tau=1$ and a canonical drift vector field of a  Gaussian noise 
    image.  
}
\end{figure}


\section{Probabilistic Diffusion Models and Osmosis}
\label{sec:discussion}

After briefly reviewing osmosis filters, we 
establish connections to probabilistic diffusion via the PDE formulation, 
highlighting similarities and differences.

\subsection{Continuous Osmosis Filtering}
\label{sec:osmosis}

Consider an initial grey value image $f : \Omega \rightarrow \mathbb{R}_+$ 
that maps coordinates from the image domain  $\Omega \subset \mathbb{R}^2$  to 
positive grey values. Osmosis describes the evolution of $u: \Omega \times 
[0,\infty) \rightarrow \mathbb{R}_+$ over time $t$, starting with $f$ at 
$t=0$.  Colour images are covered by channel-wise processing. Besides the 
initial image, the so-called \emph{drift vector field} $\bm d: 
\Omega \rightarrow \mathbb{R}^2$ has a decisive impact on the image evolution 
of $u(\bm x,t)$ which is determined by the PDE~\cite{WHBV13}
\begin{align}
    \partial_t u &\;=\; \Delta u - \Div(\bm d u) & \quad \text{on } 
    \Omega 
    \times (0,T] \, .  \label{eq:ospde}
\end{align}
Reflecting boundary conditions prevent transport across the image boundaries.
For $\bm d = \bm 0$, Eq.~\eqref{eq:ospde} describes a homogeneous 
diffusion process \cite{Ii62}, 
which smoothes the image over time. The drift component  
$-\Div(\bm d u) = -\partial_x (d_1 u) - \partial_y (d_2 u)$  
specifies local asymmetries in exchange of data between pixel cells. This makes 
it a valuable design tool for filters in visual computing.

The PDE formulation of Feller~\cite{Fe49} for probabilistic diffusion from 
Eq.~\eqref{eq:fellerfwd} closely resembles the osmosis 
equation~\eqref{eq:ospde}. Consider the 1-D version of the osmosis 
equation~\eqref{eq:ospde}:
\begin{align}
    \partial_t u &\;=\;  \partial_{xx} u - \partial_x (d u)  & \quad 
    \text{on } 
    \Omega \subset{\R}
    \times (0,T] \, .  \label{eq:1dospde}
\end{align} 
A structural comparison to Eq.~\eqref{eq:fellerfwd} reveals that the evolving image 
$u$ corresponds to the evolving probability density~$p$. In both cases, $\partial_t$ 
is the derivative w.r.t. the time variable of the image evolution. The spatial 
derivative 
$\partial_x$ corresponds to the derivative $\partial_{\bm u_t}$ w.r.t. positions of 
individual particles in the probabilistic diffusion model (see 
Section~\ref{sec:differences}). Hence, we can identify $m_2$ with a diffusivity which 
is set to one in the linear osmosis equation. The moment 
$m_1$ 
corresponds to the scalar drift term $d$ in Eq.~\eqref{eq:1dospde}.
Overall, we can interpret probabilistic diffusion as a 1-D osmosis PDE on the 
probability density $p$. Thus, there is a close conceptual connection between 
both methods. 

In this paper, we consider only those osmosis properties that are most 
relevant for our comparison to probabilistic diffusion. For the more  
theoretical details we refer to Weickert et al.~\cite{WHBV13} in the continuous 
and Vogel et al.~\cite{VHWS13} in the discrete setting. Additional theoretical 
results and their proofs can be found in the Ph.D. thesis of  
Schmidt~\cite{Sc18}. 

\subsection{Visual Comparison}

In the following, we compare structural properties of osmosis filtering and 
probabilistic diffusion models. In order to illustrate our observations, we provide a 
visual comparison of probabilistic diffusion to an osmosis evolution in 
Fig.~\ref{fig:scalespaces}. Since the probabilistic diffusion model defines an 
evolution of probability densities, the visual comparison considers a single, exemplary
trajectory. It acts as a representative for the effects on the level of individual 
images.

Such a trajectory of the probabilistic diffusion model can be directly 
obtained from the forward process. It is straightforward to implement following the 
update scheme from Eq.~\eqref{eq:forwardrandomvar}, which is already formulated in the 
discrete setting.

We discretise the continuous osmosis model from Section~\ref{sec:osmosis} 
in the same way as Vogel et al.~\cite{VHWS13}. In particular, we use an implicit 
scheme with a stabilised BiCGSTAB solver \cite{Me15}. In order to obtain the osmosis 
evolution, we use a standard normal noise sample 
$v$ to define the so-called canonical drift vector field $\bm d = \frac{\grad 
    v}{v}$. Weickert et al.~\cite{WHBV13} have shown that this yields a steady 
state $w=\frac{\mu_f}{\mu_v} v$ where $\mu_u$ denotes the average grey value of 
an image $u$. 

Note that our visual comparison is not intended as a full-scale, systematic evaluation 
of both approaches, which is beyond the scale of this paper.

\subsection{Common Structural Properties}
\label{sec:common}

Due to the observation that both osmosis and probabilistic diffusion rely on 
drift-diffusion, they share some theoretical properties and yield similar image 
evolutions in Fig.~\ref{fig:scalespaces}. 
Starting with the same initial image $f$ from the Berkeley segmentation
dataset \emph{BSDS500}~\cite{AMFM11} both processes transition to 
noise. 

We observe that a comparable amount of information of the initial image is removed 
over time. Schmidt~\cite{Sc18} has shown that Lyapunov functionals can be 
defined for osmosis. In particular, the relative entropy of $u$ w.r.t. $w$,  
\begin{align}
    L(t) \;:=\; -\int_\Omega  u(\bm x,t)\ln \left( \frac{u(\bm x,t)}{w(\bm 
    x)}\right) 
    \, 
    d\bm x \, ,
\end{align}
is increasing in $t$. It indicates that the information w.r.t. to the 
steady state is increasing. This conceptually resembles the conditional entropy 
from Section~\ref{sec:probdiffscale}, which is formulated from the point of 
view of the initial distribution instead.

Finally, the backward process for probabilistic diffusion has a conceptual 
counterpart in osmosis. By swapping  the roles of initial 
and guidance image, osmosis can also transition from noise to an image. 
However, this would not yield the same intermediate scales as the osmosis 
evolution in  Fig.~\ref{fig:scalespaces}.


\subsection{Differences}
\label{sec:differences}

In contrast to probabilistic diffusion, osmosis is deterministic and thus has 
different applications. Osmosis is applied to individual images only, but also 
does not require any training data to perform tasks like image editing 
\cite{DMM18,VHWS13,WHBV13} or shadow removal \cite{DMM18,PCCSW19,WHBV13}. 
Due to its stochastic nature, probabilistic diffusion is a natural choice for 
generating new images from given user prompts such as text. However, it can 
also be used for the restoration of individual images, e.g. restoring large 
missing image parts with matching generated content~\cite{RBLE+22,SDME21}.

The second major difference lies in the physical interpretation of images. 
Osmosis directly models the macroscopic aspects of propagation in the 2-D image 
domain $\Omega$, where pixel values can be seen as concentrations in the 
respective area. In contrast, the implementation of probabilistic diffusion via 
a Markov process models individual particles with 1-D Brownian motion. The 
pixel values are thus positional data instead, and propagation occurs in the 
co-domain. Effects of these conceptual differences are visible in 
Fig.~\ref{fig:scalespaces}. Osmosis gradually reduces the features of the cat 
by smoothing in the two-dimensional image domain. In contrast, the independent 
pixel-wise Brownian motion of probabilistic diffusion keeps edges of the 
original image intact until they are drowned out by noise.

Considering the steady state, we see that osmosis preserves the average colour 
value of the initial image, while samples from the standard normal distribution 
in probabilistic diffusion always have zero mean. Note that for osmosis, we 
only 
receive a noise steady state since we specified the guidance image accordingly 
-- we can also use arbitrary other images to guide osmosis to receive 
non-trivial steady states. In contrast, probabilistic diffusion in the sense of 
Sohl-Dickstein et al.~\cite{SMDH13} always converges to a tractable 
(noise) distribution, e.g. the standard normal distribution.  


\section{Conclusions and Outlook}
\label{sec:conclusion}

Investigating probabilistic diffusion from the point of view of scale-space 
research yields surprising results: Probabilistic diffusion defines an 
evolution that resembles traditional scale-spaces in important aspects such as 
causality and gradual simplification. However, probabilistic diffusion acts on 
distributions rather than on individual images and removes information by 
creating chaos instead of uniformity. Thus, it does not converge to a flat 
steady state, but to a noise distribution. Theoretical and practical results 
allow bi-directional 
traversal of this scale-space, which is rare in deterministic scale-spaces.

Interestingly, probabilistic diffusion can be seen as the stochastic 
counterpart to classical PDE-based, deterministic osmosis filtering. In the 
future, we plan to investigate this connection in more detail. Moreover, 
recognising probabilistic diffusion as a scale-space implies potential 
applications that make use of intermediate results of the evolution instead of 
only relying on steady states.

\medskip
\noindent
\textbf{Acknowledgements:} I thank Joachim Weickert, Karl Schrader, and 
Kristina Schaefer for fruitful discussions and advice.

\medskip
%
%
%
\bibliographystyle{splncs04}
\bibliography{bib.bib}

\begin{thebibliography}{10}
\providecommand{\url}[1]{\texttt{#1}}
\providecommand{\urlprefix}{URL }
\providecommand{\doi}[1]{https://doi.org/#1}

\bibitem{AGLM93}
Alvarez, L., Guichard, F., Lions, P.L., Morel, J.M.: Axioms and fundamental
  equations in image processing. Archive for Rational Mechanics and Analysis
  \textbf{123},  199--257 (Sep 1993)

\bibitem{AMFM11}
Arbelaez, P., Maire, M., Fowlkes, C., Malik, J.: Contour detection and
  hierarchical image segmentation. IEEE Transactions on Pattern Analysis and
  Machine Intelligence  \textbf{33}(5),  898--916 (Aug 2011)

\bibitem{DMM18}
d'Autume, M., Morel, J.M., Meinhardt-Llopis, E.: A flexible solution to the
  osmosis equation for seamless cloning and shadow removal. In: Proc~2018
  {IEEE} International Conference on Image Processing. pp. 2147--2151. Athens,
  Greece (Oct 2018)

\bibitem{DFGH04}
Duits, R., Florack, L., {de Graaf}, J., {ter Haar Romeny}, B.: On the axioms of
  scale space theory. Journal of Mathematical Imaging and Vision  \textbf{20},
  267--298 (May 2004)

\bibitem{Fe49}
Feller, W.: On the theory of stochastic processes, with particular reference to
  applications. In: First Berkeley Symposium on Mathematical Statistics and
  Probability. pp. 403--432. Berkeley, CA (Jan 1949)

\bibitem{Ga85}
Gardiner, C.W.: Handbook of Stochastic Methods for Physics, Chemistry and the
  Natural Sciences, Springer Series in Synergetics, vol.~13. Springer, Berlin
  (1985)

\bibitem{GPMX14}
Goodfellow, I.J., Pouget-Abadie, J., Mirza, M., Xu, B., Warde-Farley, D.,
  Ozair, S., Courville, A.C., Bengio, Y.: Generative adversarial nets. In:
  Ghahramani, Z., Welling, M., Cortes, C., Lawrence, N.D., Weinberger, K.Q.
  (eds.) Proc.~28th International Conference on Neural Information Processing
  Systems. Advances in Neural Information Processing Systems, vol.~27, pp.
  2672--2680. Montr\'eal, Canada (Dec 2014)

\bibitem{SHH23}
Hagemann, P.L., Hertrich, J., Steidl, G.: Generalized normalizing flows via
  {M}arkov chains. In: Elements in Non-local Data Interactions: Foundations and
  Applications. Cambridge University Press (2023), in press

\bibitem{HJA20}
Ho, J., Jain, A., Abbeel, P.: Denoising diffusion probabilistic models. In:
  Advances in Neural Information Processing Systems, vol.~33, pp. 6840--6851.
  NeurIPS Foundation, San Diego, CA (2020)

\bibitem{HKMR+16}
Huckemann, S., Kim, K.R., Munk, A., Rehfeldt, F., Sommerfeld, M., Weickert, J.,
  Wollnik, C.: The circular sizer, inferred persistence of shape parameters and
  application to early stem cell differentiation. Bernoulli  \textbf{22}(4),
  2113--2142 (Nov 2016)

\bibitem{Ii62}
Iijima, T.: Basic theory on normalization of pattern (in case of typical
  one-dimensional pattern). Bulletin of the Electrotechnical Laboratory
  \textbf{26},  368--388 (Jan 1962), in Japanese

\bibitem{KSBH21}
Kingma, D., Salimans, T., Poole, B., Ho, J.: Variational diffusion models.
  Advances in neural information processing systems  \textbf{34},  21696--21707
  (2021)

\bibitem{KV99}
Koenderink, J.J., {Van Doorn}, A.J.: The structure of locally orderless images.
  International Journal of Computer Vision  \textbf{31}(2),  159--168 (Apr
  1999)

\bibitem{Li11}
Lindeberg, T.: Generalized {G}aussian scale-space axiomatics comprising linear
  scale-space, affine scale-space and spatio-temporal scale-space. Journal of
  Mathematical Imaging and Vision  \textbf{40},  36--81 (2011)

\bibitem{Ma00}
Majer, P.: A statistical approach to feature detection and scale selection in
  images. Ph.D. thesis, Department of Mathematics, Saarland University,
  G\"ottingen, Germany (2000)

\bibitem{Me15}
Meister, A.: Numerik linearer Gleichungssysteme. Vieweg, Braunschweig, 5th edn.
  (2015)

\bibitem{PCBP+20}
Parisotto, S., Calatroni, L., Bugeau, A., Papadakis, N., Sch{\"o}nlieb, C.B.:
  Variational osmosis for non-linear image fusion. IEEE Transactions on Image
  Processing  \textbf{29},  5507--5516 (Apr 2020)

\bibitem{PCCSW19}
Parisotto, S., Calatroni, L., Caliari, M., , Sch{\"o}nlieb, C.B., Weickert, J.:
  Anisotropic osmosis filtering for shadow removal in images. Inverse Problems
  \textbf{35}(5) (Apr 2019), article 054001

\bibitem{Pe03}
Pedersen, K.S.: Properties of {B}rownian image models in scale-space. In:
  Griffin, L.D., Lillholm, M. (eds.) Scale-Space Methods in Computer Vision,
  Lecture Notes in Computer Science, vol.~2695, pp. 281--296. Springer, Berlin
  (2003)

\bibitem{Ri84}
Risken, H.: The Fokker--Planck Equation. Springer, New York (1984)

\bibitem{RBLE+22}
Rombach, R., Blattmann, A., Lorenz, D., Esser, P., Ommer, B.: High-resolution
  image synthesis with latent diffusion models. In: Proc.~2022 {IEEE}/{CVF}
  Conference on Computer Vision and Pattern Recognition. pp. 10684--10695. New
  Orleans, LA (Jun 2022)

\bibitem{SchW98}
Scherzer, O., Weickert, J.: Relations between regularization and diffusion
  filtering. Journal of Mathematical Imaging and Vision  \textbf{12}(1),
  43--63 (Feb 2000)

\bibitem{Sc18}
Schmidt, M.: Linear Scale-Spaces in Image Processing: Drift–Diffusion and
  Connections to Mathematical Morphology. Ph.D. thesis, Department of
  Mathematics, Saarland University, Saarbr\"ucken, Germany (2018)

\bibitem{SW16}
Schmidt, M., Weickert, J.: Morphological counterparts of linear shift-invariant
  scale-spaces. Journal of Mathematical Imaging and Vision  \textbf{56}(2),
  352--366 (Oct 2016)

\bibitem{So01b}
Sochen, N.A.: Stochastic processes in vision: {F}rom {L}angevin to {B}eltrami.
  In: Proc.~Eighth International Conference on Computer Vision. vol.~1, pp.
  288--293. Vancouver, Canada (Jul 2001)

\bibitem{SMDH13}
Sohl-Dickstein, J., Weiss, E., Maheswaranathan, N., Ganguli, S.: Deep
  unsupervised learning using nonequilibrium thermodynamics. In: Bach, F.,
  Blei, D. (eds.) Proc.~32nd International Conference on Machine Learning.
  Proceedings of Machine Learning Research, vol.~37, pp. 2256--2265. Lille,
  France (Jul 2015)

\bibitem{SDME21}
Song, Y., Durkan, C., Murray, I., Ermon, S.: Maximum likelihood training of
  score-based diffusion models. Advances in Neural Information Processing
  Systems  \textbf{34},  1415--1428 (2021)

\bibitem{VHWS13}
Vogel, O., Hagenburg, K., Weickert, J., Setzer, S.: A fully discrete theory for
  linear osmosis filtering. In: Kuijper, A., Bredies, K., Pock, T., Bischof, H.
  (eds.) Scale Space and Variational Methods in Computer Vision, Lecture Notes
  in Computer Science, vol.~7893, pp. 368--379. Springer, Berlin (2013)

\bibitem{We97}
Weickert, J.: Anisotropic Diffusion in Image Processing. Teubner, Stuttgart
  (1998)

\bibitem{WHBV13}
Weickert, J., Hagenburg, K., Breu{\ss}, M., Vogel, O.: Linear osmosis models
  for visual computing. In: Heyden, A., Kahl, F., Olsson, C., Oskarsson, M.,
  Tai, X.C. (eds.) Energy Minimisation Methods in Computer Vision and Pattern
  Recognition, Lecture Notes in Computer Science, vol.~8081, pp. 26--39.
  Springer, Berlin (2013)

\end{thebibliography}

\end{document}